\documentclass{sig-alternate-05-2015}

\usepackage{enumitem}
\setitemize{noitemsep,topsep=0pt,parsep=0pt,partopsep=0pt}
\usepackage{graphicx}
\graphicspath{{figures/}}
\usepackage[colorlinks]{hyperref}
\usepackage{amsmath,amssymb,url,listings,mathrsfs,wasysym}
\usepackage{pgf}
\usepackage{tikz}
\usetikzlibrary{shadows}
\usetikzlibrary{shapes}
\usetikzlibrary{arrows,automata,positioning}
\usepackage{multirow}
\usepackage{algpseudocode}
\usepackage{algorithm}
\usepackage{cite}

\newcommand{\floor}[1]{\lfloor #1 \rfloor}

\algtext*{EndWhile}
\algtext*{EndIf}
\algtext*{EndFunction}
\algtext*{EndFor}

\newtheorem{coro}{Corollary}
\newtheorem{remark}{Remark}

\newtheorem{example}{Example}

\newtheorem{definition}{Definition}

\newtheorem{thm}{Theorem}
\newtheorem{prop}{Proposition}
\newtheorem{ass}{Assumption}
\newtheorem{lem}{Lemma}
\newtheorem{rem}{Remark}

\newcommand{\abs}[1]{\left\vert#1\right\vert}
\newcommand{\set}[1]{\left\{#1\right\}}
\newcommand{\Real}{\mathbb R}
\newcommand{\eps}{\varepsilon}
\renewcommand{\phi}{\varphi}

\renewcommand{\subset}{\subseteq}

\newcommand{\U}{\mathcal{U}}

\newcommand{\T}{\mathcal{T}}

\renewcommand{\S}{\mathcal S}

\newcommand{\ra}{\rightarrow}

\newcommand{\Z}{\mathbb{Z}}
\newcommand{\ltlx}{\text{LTL}}

\newcommand{\true}{\relax\ifmmode \mathit{True} \else \em True \/\fi}
\newcommand{\false}{\relax\ifmmode \mathit{False} \else \em False \/\fi}
\newcommand{\eventually}{\Diamond}
\newcommand{\always}{\Box}

\makeatletter
\def\@copyrightspace{\relax}
\makeatother

\begin{document}

\title{Robust Abstractions for Control Synthesis: \\Robustness Equals Realizability for Linear-Time Properties}

%
%
%
%
%

\numberofauthors{1} 
%
\author{
%
%
\alignauthor
Jun Liu\\[1\jot]
      \affaddr{Department of Applied Mathematics}\\[1\jot]
       \affaddr{University of Waterloo}\\[1\jot]
       \email{j.liu@uwaterloo.ca}
}

\maketitle
\begin{abstract}
We define robust abstractions for synthesizing provably correct and robust controllers for (possibly infinite) uncertain  transition systems. It is shown that robust abstractions are sound in the sense that they preserve robust satisfaction of linear-time properties. We then focus on discrete-time control systems modelled by nonlinear difference equations with inputs and define concrete robust abstractions for them. While most abstraction techniques in the literature for nonlinear systems focus on constructing sound abstractions, we present computational procedures for constructing both sound and approximately complete robust abstractions for general nonlinear control systems without stability assumptions. Such procedures are approximately complete in the sense that, given a concrete discrete-time control system and an arbitrarily small perturbation of this system, there exists a finite transition system that robustly abstracts the concrete system and is abstracted by the slightly perturbed system simultaneously. A direct consequence of this result is that robust control synthesis for discrete-time nonlinear systems and linear-time specifications is robustly decidable. More specifically, if there exists a robust control strategy that realizes a given linear-time specification, we can algorithmically construct a (potentially less) robust control strategy that realizes the same specification. The theoretical results are illustrated with a simple motion planning example. 
\end{abstract}

\printccsdesc


\keywords{Nonlinear systems; control synthesis; abstraction; robustness; linear-time property; linear temporal logic; decidability}

\section{Introduction}

Abstraction serves as a bridge for connecting control theory and formal methods in the sense that hybrid control design for dynamical systems and high-level specifications can be done using finite abstractions of these systems \cite{alur2000discrete,tabuada2009verification}. There has been a rich literature on computing abstractions for linear and nonlinear dynamical systems in the past decade (see, e.g., \cite{girard2010approximately,liu2016finite,liu2014abstraction,pola2008approximately,reiszig2009computing,tazaki2012discrete,zamani2012symbolic}). Early work on abstraction focuses on constructing  symbolic models that are bisimilar (equivalent) to the original system. The seminal work in \cite{tabuada2006linear} shows that bisimilar symbolic models exist for controllable linear systems. As a result, existence of controllers for such systems to meet linear-time properties (such as those specified by linear temporal logic \cite{baier2008principles}) is decidable. For nonlinear systems that are incrementally stable \cite{angeli2002lyapunov}, it is shown in \cite{pola2008approximately} that approximately bisimilar models can be constructed (see also \cite{girard2010approximately}, for construction of approximately bisimilar models for switched systems, and \cite{girard2012controller} for its use in control synthesis). The work in \cite{zamani2012symbolic} considered symbolic models for nonlinear systems without stability assumptions, in which it is shown that symbolic models that approximately alternatingly simulate the sample-data representation of a general nonlinear control system can be constructed. The work in \cite{reiszig2009computing} and \cite{tazaki2012discrete} both proposes computational procedures for constructing finite abstractions of discrete-time nonlinear systems. The abstraction techniques in \cite{zamani2012symbolic,reiszig2009computing,tazaki2012discrete} are conservative and sound in the sense that they are useful in the design of provably correct controllers, but do not necessarily yield a feasible design because the computational procedures for constructing abstractions for potentially unstable nonlinear systems are not complete. 

Robustness is a central property to consider in control design, because all practical control systems need to be robust to imperfections in all aspects of control design and implementation, such as modelling, sensing, computation, communication, and actuation. For abstraction-based control design, how to preserve robustness poses a particular challenge because the hierarchical control design approach based on abstraction often use quantized state measurements (modelled as symbolic states in the abstraction) to compute appropriate control signals. Because of the state quantizers by definition are discontinuous, special attention is required to ensure that the resulting design is actually robust to measurement errors and disturbances. The work in \cite{liu2014abstraction} (see also \cite{liu2016finite}) proposes a novel notion of abstractions that are equipped with additional robustness margins to cope with different types of uncertainties in modelling, such as measurement errors, delays, and disturbances. The work in \cite{reissig2014feedback} (see also \cite{reissig2016feedback}) defines a new notation of system relations for abstraction-based control design. By explicitly considering the interconnection of state quantizers and feedback controllers, it is shown that the new system relation can also be used to design robust controllers against uncertainties and disturbances. The type of abstractions considered in \cite{liu2014abstraction,liu2016finite,reissig2014feedback,reissig2016feedback} resemble the approximate alternating simulations considered in \cite{zamani2012symbolic} for nonlinear systems. These abstractions, nonetheless, are all conservative and sound. To the best knowledge of the authors, how to compute complete abstractions (or approximately complete) abstractions for general nonlinear systems without stability assumptions remains an open problem. 

As an attempt to bridge this gap, in this paper, we define robust abstractions as a system relation from a (possibly infinite) transition system subject to uncertainty to anther transition system. We show that, while this abstraction relation is to some extent similar to the type of system relations considered in \cite{liu2016finite,reissig2016feedback,zamani2012symbolic}, it also has some subtle differences that are important for proving the approximate completeness results later in the paper. We show that robust abstractions are sound in the sense that they preserve robust satisfaction of linear-time properties. The main contributions of the paper include computational procedures for constructing both sound and approximately complete robust abstractions for general discrete-time nonlinear control systems without stability assumptions. We show that such procedures are complete in the sense that, given a concrete discrete-time control system and an arbitrarily small perturbation of this system, there exists a finite transition system that robustly abstracts the concrete system, whereas the perturbed system abstracts this finite transition system. An important consequence of this main result asserts that existence of robust controllers for discrete-time nonlinear systems and linear-time specifications is decidable. Finally, we would like to make clear upfront that the main point of this paper is not on providing more efficient algorithms for computing abstractions. Therefore, complexity issues, though important, are not a concern for the current paper and will be investigated in future work. 

The organization of the paper is very straightforward. Section \ref{sec:ts} presents some background material on transition systems and define robust abstractions. We highlight some similarities and subtle differences of the new abstraction relation with several variants of simulation relations in the literature. Section \ref{sec:main} presents the main results of the paper on construction of sound and approximately complete robust abstractions for discrete-time nonlinear control systems. A numerical example is used to illustrate the effectiveness of robust abstractions in Section \ref{sec:ex}. The paper is concluded in Section \ref{sec:conc}.  

\textbf{Notation:} Let $f$ be a (binary) relation from $A$ to $B$, i.e., $f$ is a subset of the Cartesian product $A\times B$. For each $a\in A$,  $f(a)$ denotes the set $\,\set{b:\,b\in B\text{ such that }(a,b)\in f}$; for each $b\in B$, $f^{-1}(b)$ denotes the set $\,\set{a:\,a\in A,\,(a,b)\in f}$; for $A'\subset A$, $f(A')=\cup_{a\in A'} f(a)$; and for $B'\subset B$, $f^{-1}(B)=\cup_{b\in B'}f^{-1}(B)$. Let $g$ be a relation from $A$ to $B$ and $f$ be a relation from $B$ to $C$. The composition of $f$ and $g$, denoted by $f\circ g$, is a relation from $A$ to $C$ defined by 
$$
f\circ g=\set{(a,c):\,\exists b\in B\text{ s.t. }(a,b)\in g\text{ and }(b,c)\in f}.
$$ 
For two sets $A,B\subset \Real^n$, 
$$A+B = \set{c:\,\exists a\in A,\exists b\in B\text{ s.t. }a+b=c}$$ and $A\backslash B=\set{a:\,a\in A,a\not\in B}$. 
For $a\in \Real^n$ and $B\subset \Real^n$, $a+B=\set{a}+B$. Let $\abs{\cdot}$ denote the infinity norm in $\Real^n$ and $\mathbb{B}$ denote the unit closed ball in infinity norm centred at the origin, i.e. $\mathbb{B}=\set{x\in\Real^n:\,\abs{x}\le 1}$. The dimension of $\mathbb{B}$ will be clear from the context.

\section{Transition systems and robust abstractions}\label{sec:ts}

\subsection{Transition systems}

\begin{definition}\em \label{def:ts}
  A \emph{transition system} is a tuple
  \begin{displaymath}
    \T=(Q,A,R,\Pi,L),
  \end{displaymath}
where
  \begin{itemize}
  \item $Q$ is the set of states; 
  \item $A$ is the set of actions;
  \item $R \subseteq Q \times A \times Q$ is the transition relation;
  \item $\Pi$ is the set of atomic propositions;
  \item $L:Q \to 2^{\Pi}$ is the labelling function.
  \end{itemize}
\end{definition}

Consider the transition system $\T$ above. For each action $a\in A$ and $q\in Q$, the $a$-successor of $q$, denoted by $\text{Post}_{\T}(q,a)$, is defined by 
$$
\text{Post}_{\T}(q,a)=\set{q':\,q'\in Q\text{ s.t. }(q,a,q')\in R}. 
$$
For each $q\in Q$, the set of admissible actions for $q$, denoted by $A_{\T}(q)$, is defined by
$$
A_{\T}(q)=\set{a:\,\text{Post}_{\T}(q,a)\neq \emptyset}. 
$$
In this paper, we assume that all transition systems have no terminal states in the sense that $A_{\T}(q)\neq \emptyset$ for all $q\in Q$. 

An \emph{execution} of $\T$ is an infinite alternating sequence of states and actions  
$$\rho=q_0a_0q_1a_1q_2a_2\cdots,$$
where $q_0$ is some initial state and $(q_i,a_i,q_{i+1})\in R$ for all $i\ge 0$. 
The \emph{path} resulting from the execution $\rho$ above is 
$$
\text{Path}(\rho)=q_0q_1q_2\cdots.
$$
The \emph{trace} of the execution $\rho$ is defined by 
$$\text{Trace}(\rho)=L(q_0)L(q_1)L(q_2)\cdots.$$
A \emph{control strategy} for a transition system $\T$ is a partial function $s:\,(q_0,q_1,\cdots,q_i)\mapsto a_i$ that maps the state history to the next action. An \emph{$s$-controlled execution} of a transition system $\T$ is an execution of $\T$, where for each $i\geq 0$, the action $a_i$ is chosen according to the control strategy $s$; $s$-controlled paths and traces are defined in a similar fashion.

\subsection{Uncertainty transition systems}

\begin{definition}\em \label{def:delta}
A transition relation $\Delta \subseteq Q \times A \times Q$ is called an \emph{uncertain transition relation} for $\T=(Q,A,R,\Pi,L)$, if the following two conditions hold:
\begin{itemize}
\item[(i)] $R\cap\Delta=\emptyset$; 
\item[(ii)] for each $(q,a,q')\in \Delta$, there exists some $(q,a,q'')\in R$. 
\end{itemize}
\end{definition}

\begin{definition}\em \label{def:uts}
An \emph{uncertain transition system} consisting of $\T=(Q,A,R,\Pi,L)$ as a nominal transition system and $\Delta$ as an uncertain transition relation for $\T$, denoted by $\T\oplus\Delta$, is defined by 
$$
\T\oplus\Delta=(Q,A,R\cup\Delta,\Pi,L). 
$$
\end{definition}

It is clear from the above definition that, while $\Delta$ introduces additional transitions for the transition system $\T$, it does not add more admissible actions for any state. In other words, for all $q\in Q$, $A_{\T}(q)=A_{\T\oplus\Delta}(q)$.  

Since an {uncertain transition system} is simply a transition system with additional transitions introduced by some uncertain transition relation, the execution (path, trace), control strategy, and controlled execution (path, trace) for an {uncertain transition system} are defined in the same way as for a nominal transition system.

\subsection{Robust abstractions}

We first define a notion of abstraction between transition systems for control synthesis. 

\begin{definition}\label{def:ra}\em 
For two transition systems 
$$\T_1=(Q_1,A_1,R_1,\Pi,L_1)$$
and 
$$\T_2=(Q_2,A_2,R_2,\Pi,L_2),$$
a relation $\alpha\subset Q_1\times Q_2$ is said to be an \emph{abstraction} from $\T_1$ to $\T_2$, if the following conditions are satisfied: 
\begin{itemize}
\item[(i)]  for all $q_1\in Q_1$, there exists $q_2\in Q_2$ such that $(q_1,q_2)\in \alpha$ (i.e., $\alpha(q_1)\neq \emptyset$); 
\item[(ii)] for all $(q_1,q_2)\in \alpha$ and $a_2\in A_{\T_2}(q_2)$, there exists $a_1\in A_{\T_1}(q_1)$   
such that 
\begin{equation}\label{eq:over}
\alpha(\text{Post}_{\T_1}(q,a_1))\subset \text{Post}_{\T_2}(q_2,a_2);
\end{equation} 
for all $q\in \alpha^{-1}(q_2)$; 
\item[(iii)] for all $(q_1,q_2)\in \alpha$, $L_2(q_2)\subset L_1(q_1)$. 
\end{itemize}
If such a relation $\alpha$ exists, we say that $\T_2$ \emph{abstracts} $\T_1$ and write $\T_1\preceq_{\alpha} \T_2$ or simply $\T_1\preceq \T_2$. 
\end{definition}

We then define robust abstractions as abstractions of uncertain transition systems. 

\begin{definition}\label{def:rra}\em 
Let $\Delta$ be an {uncertain transition relation} for $\T_1$. If there exists an abstraction $\alpha$ from $\T_1\oplus\Delta$ to $\T_2$, i.e., $\T_1\oplus\Delta\preceq_{\alpha} \T_2$, we say that $\alpha$ is a $\Delta$-robust abstraction from $\T_1$ to $\T_2$ and $\T_2$ \emph{$\Delta$-robustly abstracts} $\T_1$. With a slight abuse of terminology, we sometimes also say that $\T_2$ is a $\Delta$-robust abstraction of $\T_1$.   
\end{definition}

\begin{remark}\em 
We highlight several differences between the notation of abstraction proposed in Definition \ref{def:ra} and other similar system relations in the literature. Apart from the obvious distinction that, in Definition \ref{def:ra}, an explicit model of the uncertainty is considered (following \cite{topcu2012synthesizing}), the abstraction defined by Definition \ref{def:ra} differs from several variants of simulation relations in the literature as elaborated below:\\

\emph{Finite abstractions with robustness margins:} This notion of abstractions introduced in \cite{liu2014abstraction,liu2016finite} is defined by introducing two positive parameters $(\gamma_1, \gamma_2)$, which define the extra transitions to be added to the  abstractions to ensure robustness. Suppose there is a metric $d$ defined on $Q_1$. Then finite abstractions with robustness margins $(\gamma_1, \gamma_2)$ amount to defining 
\begin{align*}
\Delta &= \{(q,a,q'):\,\exists (q_1,a,q_1')\in R_1\text{ s.t. }\\
&\qquad\qquad\qquad\qquad  d(q_1,q)\le \gamma_1,\,d(q_1',q')\le \gamma_2\}\backslash R_1. 
\end{align*}
To establish $\T_1\oplus\Delta\preceq_{\alpha} \T_2$, condition (\ref{eq:over}), which can be equivalently written as
\begin{align*}
\bigcup_{q\in \alpha^{-1}(q_2)}\alpha(\text{Post}_{\T_1\oplus\Delta}(q,a_1)) &=\alpha(\bigcup_{q\in \alpha^{-1}(q_2)}\text{Post}_{\T_1\oplus\Delta}(q,a_1)) \\
&\subset \text{Post}_{\T_2}(q_2,a_2)
\end{align*}
is essentially the over-approximation (of transitions) condition in \cite{liu2014abstraction,liu2016finite}. The main difference lies in that Definition \ref{def:ra} does not assume that a metric is defined on $Q_1$ and the uncertainty model is not restricted to that defined by level sets of the distance function. Furthermore, here we define the abstraction relation on a general Kripke structure, whereas the work in \cite{liu2014abstraction,liu2016finite} defines concrete abstractions from ordinary differential/difference equations with inputs to finite transition systems. 

\emph{Feedback refinement relations \cite{reissig2014feedback,reissig2016feedback}:} Similar to \cite{liu2014abstraction,liu2016finite}, the abstraction relation considered in \cite{reissig2014feedback,reissig2016feedback} also requires that, for each $(q_1,q_2)\in \alpha$, the admissible actions for each $q_2$ is a subset of the admissible actions for $q_1$. In Definition \ref{def:ra}, for each $(q_1,q_2)\in \alpha$, it is not required that $A_{\T_2}(q_2)\subset A_{\T_1}(q_1)$, i.e., the admissible actions for $q_1$ do not have to be a subset of the admissible actions for $q_2$. This difference enables us to formulate and prove the approximate completeness results later in this paper (Section \ref{sec:complete}). Note that, when $A_{\T_2}(q_2)\subset A_{\T_1}(q_1)$, condition (\ref{eq:over}) can be simplified to: for each $(q_1,q_2)\in\alpha$ and every $a\in A_{\T_2}(q_2)$, 
\begin{equation}\label{eq:frr}
\alpha(\text{Post}_{\T_1}(q_1,a))\subset \text{Post}_{\T_2}(q_2,a).
\end{equation} 
In other words, the same action $a$ used by $q_2$ is assumed to be available (and used) for all $q_1\in\alpha^{-1}(q_2)$, because $A_{\T_2}(q_2)\subset A_{\T_1}(q_1)$.

\emph{Alternating simulations \cite{pola2008approximately,zamani2012symbolic}:} The notion of alternating simulations \cite{pola2008approximately,zamani2012symbolic} 
stipulates that, for each $(q_1,q_2)\in\alpha$ and every $a_2\in A_{\T_2}(q_2)$, there exists $a_1\in A_{\T_2}(q_1)$ such that, for every $q_1'\in\text{Post}_{\T_1}(q_1,a_1)$, there exists {some state} $q_2'\in \text{Post}_{\T_2}(q_2,a_2)$ such that $(q_1',q_2')\in\alpha$. In other words, for each $(q_1,q_2)\in\alpha$ and every $a_2\in A_{\T_2}(q_2)$, there exists $a_1\in A_{\T_1}(q_1)$ such that 
\begin{equation}\label{eq:as}
\alpha(q_1')\cap \text{Post}_{\T_2}(q_2,a_2) \neq \emptyset,
\end{equation} 
for all $q_1'\in \text{Post}_{\T_1}(q_1,a_1)$, as articulated in \cite{reissig2014feedback,reissig2016feedback}. Clearly, (\ref{eq:as}) is a weaker condition than (\ref{eq:over}) or (\ref{eq:frr}), unless $\alpha$ is single-valued. Furthermore, and more importantly, (\ref{eq:as}) does not stipulate the use of the same action $a_1$ for all $q\in\alpha^{-1}(q_2)$, i.e., $a_1$ may depend on $q$ (concrete states corresponding to $q_2$). A consequence of the latter is that, to implement the controller, one needs knowledge of the concrete state rather than the abstract (symbolic) state alone. 
\end{remark}

We use a simple example to illustrate the differences discussed above.  

\begin{example}
Consider three transition systems 
$$\T_i=(Q_i,A_i,R_i,\Pi,L_i),\quad i=1,2,3,$$ 
where $Q_1=\set{x_0,x_1,x_2}$, $Q_2=Q_3=\set{q_0,q_1}$, $A_1=\set{a,b}$, $A_2=A_3=\set{1,2,3}$, $\Pi=\set{\text{Initial}, \text{Goal}}$, $L_1(x_0)=L_1(x_1)=L_2(q_0)=L_3(q_0)=\set{\text{Initial}}$, and $L_1(x_2)=L_2(q_1)=L_3(q_1)=\set{\text{Goal}}$. The transition relations are shown in Figure \ref{fig:fts1}. 
\begin{figure}[ht!]
    \centering
    \includegraphics[width=0.47\textwidth]{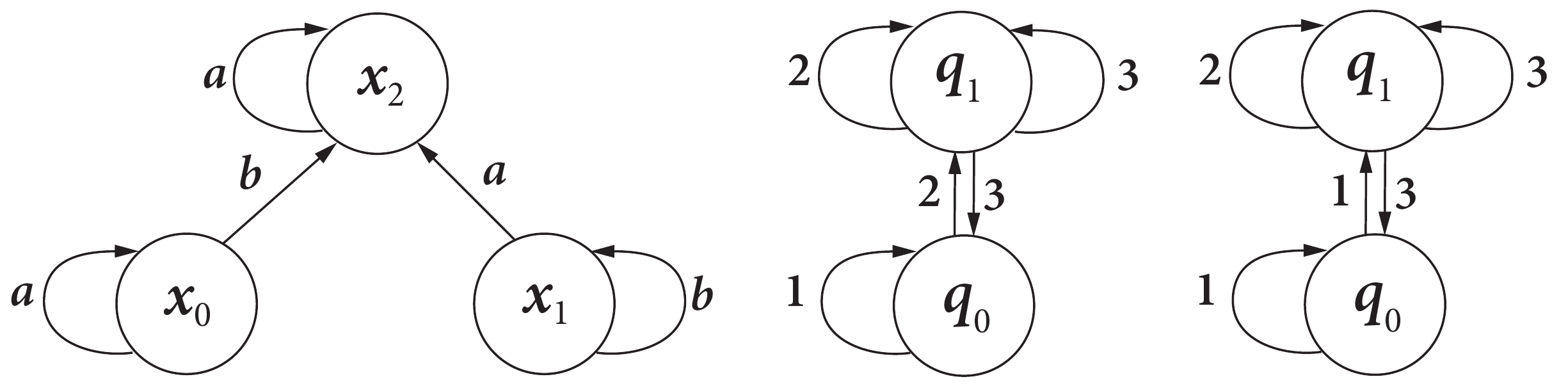}
    \caption{Transition systems $\T_1$ (left), $\T_2$ (middle), and $\T_3$ (right).}
    \label{fig:fts1}
\end{figure}

Define an abstraction relation from $\T_1$ to $\T_2$ by 
\begin{equation*}
\alpha=\set{(x_0,q_0),(x_1,q_0),(x_2,q_1)}. 
\end{equation*}
Then it can be easily verified that (\ref{eq:as}) is satisfied and $\alpha$ is an alternating simulation from $\T_1$ to $\T_2$. In fact, we can check that, for $(x_0,q_0)\in\alpha$, and action $1\in A_2$, there exists $a\in A_1$ such that 
$$\alpha(\text{Post}_{\T_1}(x_0,a))=\alpha(\set{x_0})=\set{q_0}=\text{Post}_{\T_2}(q_0,1),$$ 
which implies (\ref{eq:as}). Similarly, for $(x_0,q_0)\in\alpha$, and action $2\in A_2$, there exists $b\in A_1$ such that 
$$\alpha(\text{Post}_{\T_1}(x_0,b))=\alpha(\set{x_2})=\set{q_1}=\text{Post}_{\T_2}(q_0,2),$$ 
which also implies  (\ref{eq:as}). 
For $(x_2,q_1)\in\alpha$, and action $3\in A_2$, there exists $a\in A_1$ such that 
$$\alpha(\text{Post}_{\T_1}(x_2,a))=\alpha(\set{x_2})=\set{q_1}\subset \set{q_0,q_1}=\text{Post}_{\T_2}(q_1,3),$$ 
which implies   (\ref{eq:as}). The rest can be checked in a similar fashion. 

Suppose that one needs to design a control strategy for $\T_1$ such that all controlled executions of $\T_1$ starting from the 'Initial' set will eventually reach the 'Goal' set. Then, while one can find such a control strategy for $\T_2$, to implement this strategy on $\T_1$, however, $\T_1$ needs to be able to discriminate $x_0$ and $x_1$ and choose the appropriate actions ($b$ for $x_0$ and $a$ for $x_1$). This is not the case if only symbolic state information from the abstraction is available.

Note that, according to Definition \ref{def:ra}, we do not have $\T_1\preceq_{\alpha}\T_2$ because, for $(x_0,q_0)\in\alpha$ and action $1\in A_2$, we have
\begin{align*}
\bigcup_{x\in\alpha^{-1}(q_0)}\alpha(\text{Post}_{\T_1}(x,a))&=\alpha(\set{x_0},\set{x_2})\\
&=\set{q_0,q_1}\not\subset\set{q_0}=\text{Post}_{\T_2}(q_0,1),
\end{align*} 
\begin{align*}
\bigcup_{x\in\alpha^{-1}(q_0)}\alpha(\text{Post}_{\T_1}(x,b))&=\alpha(\set{x_1},\set{x_2})\\
&=\set{q_0,q_1}\not\subset\set{q_0}=\text{Post}_{\T_2}(q_0,1),
\end{align*} 

Thus, (\ref{eq:over}) does not hold for either action $a$ or $b$.

We can check that  $\T_1\preceq_{\alpha}\T_3$. Because the set of actions in $\T_2$ (and $\T_3$) is not a subset of the actions of $\T_1$ (in fact there are more actions in $\T_2$ and $\T_3$ than $\T_1$), $\alpha$ does not provide an abstraction relation from $\T_1$ to $\T_2$ or from $\T_1$ to $\T_3$ in the strict sense of the notions of simulation relations considered in \cite{liu2014abstraction,liu2016finite,reissig2014feedback,reissig2016feedback}. 

To consider a robust abstraction for $\T_1$, let $\Delta = \set{(x_2,a,x_1)}$. Then it can be verified that the transition system $\T_3$ is also a $\Delta$-robust abstraction of $\T_1$. 
\end{example}

We will state some immediate results that follow from Definition \ref{def:ra}. 

\begin{prop}\label{prop:well}
Let $\T$ be a transition system and $\Delta$ be an uncertain transition relation for $\T$. Then $\T\preceq\T\oplus\Delta$. 
\end{prop}

\begin{proof}
Let  $\T=(Q,A,R,\Pi,L)$. It is straightforward to check by Definitions \ref{def:delta} and \ref{def:ra} that the identity relation from $Q$ to $Q$  defines a $\Delta$-robust abstraction from $\T$ to $\T\oplus\Delta$. 
\end{proof}

Setting $\Delta=\emptyset$, a special case of Proposition \ref{prop:well} asserts that $\T\preceq \T$ for any transition system. It is also straightforward to verify that abstraction relations are transitive in the following sense. 

\begin{prop}\label{prop:trans}
Let $\T_i$ ($i=1,2,3$) be transition systems and $\Delta$ be an uncertain transition relation for $\T_1$. If 
$\T_1\,\preceq_{\alpha_1} \T_2$ and $\T_2\,\preceq_{\alpha_2}\T_3$, then $\T_1\,\preceq_{\alpha_2\circ\alpha_1} \T_3$. 
\end{prop}

\begin{proof}
Let $\alpha_3=\alpha_2\circ \alpha_1$. We verify that conditions (i)--(iii) of Definition \ref{def:ra} are satisfied:
\begin{itemize}
\item[(i)] For all $q_1\in Q$, $\alpha_3(q_1)$ is non-empty, because $\alpha_1(q_1)$ is non-empty and $\alpha_2(q_2)$ is non-empty for any $q_2\in Q_2$.\\
\item[(ii)] {\color{black} For any $(q_1,q_3)\in \alpha_3$}, there exists $q_2\in Q_2$ such that $(q_1,q_2)\in \alpha_1$ and $(q_2,q_3)\in \alpha_2$. For any $q_3\in A_{\T_3}(q_3)$, there exists $a_2\in A_{\T_2}(q_2)$ such that 
$$
\alpha_2(\text{Post}_{\T_2}(q,a_2))\subset \text{Post}_{\T_3}(q_3,a_3),
$$
for all $q\in \alpha_2^{-1}(q_3)$.  For $a_2\in A_{\T_2}(q_2)$, there exists $a_1\in A_{\T_1}(q_1)$ such that 
$$
\alpha_1(\text{Post}_{\T_1}(q,a_1))\subset \text{Post}_{\T_2}(q_2,a_2),
$$
for all $q\in \alpha_1^{-1}(q_2)$. It follows that 
\begin{align*}
&\bigcup_{q\in \alpha_3^{-1}(q_3)}\alpha_3(\text{Post}_{\T_1}(q,a_1))\\
&\qquad=\bigcup_{q\in \alpha_3^{-1}(q_3)}\alpha_2\circ \alpha_1(\text{Post}_{\T_1}(q,a_1))\\
&\qquad\subset \bigcup_{q\in \alpha_2^{-1}(q_2)}\alpha_2(\text{Post}_{\T_2}(q,a_2))\\
&\qquad\subset \text{Post}_{\T_3}(q_3,a_3). 
\end{align*}
\item[(iii)] For any $(q_1,q_3)\in \alpha_3$, there exists $q_2\in Q_2$ such that $(q_1,q_2)\in \alpha_1$ and $(q_2,q_3)\in \alpha_2$. Hence
$$
L_3(q_3)\subset L_2(q_2)\subset L_1(q_1). 
$$
\end{itemize}
\end{proof}

\subsection{Soundness of abstractions}

In this section, we prove that abstractions given by Definition \ref{def:ra} are sound in the sense of preserving realizability of linear-time properties. 

A \emph{linear-time (LT) property} \cite{baier2008principles} over a set of atomic propositions $\Pi$ is a subset of $(2^\Pi)^\omega$, which is the set of all infinite words over the alphabet $2^\Pi$, defined by
$$
(2^\Pi)^\omega= \set{A_0A_1A_2\cdots:\,A_i\in 2^\Pi,\quad i\ge 0}. 
$$
A particular class of LT properties can be conveniently specified by \emph{linear temporal logic} (LTL \cite{pnueli1977temporal}). This logic consists of propositional logic operators (e.g., \textbf{true}, \textbf{false}, {\em negation\/} ($\neg$), {\em disjunction\/} ($\vee$), {\em conjunction\/} ($\wedge$) and
{\em implication\/} ($\rightarrow$)), and temporal operators (e.g., {\em next\/} ($\bigcirc$), {\em always\/} ($\always$),
{\em eventually\/} ($\eventually$), {\em until\/} ($\mathcal{U}$) and {\em weak until\/} ($\mathcal{W}$)).

The syntax of LTL over a set of atomic propositions $\Pi$ is defined inductively follows:
\begin{itemize}
\item \textbf{true} and \textbf{false} are $\ltlx$ formulae; 
\item an atomic proposition $\pi\in\Pi$ is an $\ltlx$ formula; 
\item if $\phi$ and $\psi$ are $\ltlx$ formulas, then $\neg\phi$, $\phi \vee \phi$, $\bigcirc\phi$, and $\phi \U \phi$ are $\ltlx$ formulas.
\end{itemize}

The semantics of $\ltlx$ is defined on infinite words over the alphabet $2^\Pi$. Given a sequence $\sigma=A_0A_1A_2\cdots$ in $2^{\Pi}$, we define $\sigma,i\vDash \phi$, meaning that $\sigma$ satisfies an $\ltlx$ formula $\phi$ at position $i$, inductively as follows:
\begin{itemize}
\item $\sigma,i\vDash \textbf{true}$; 
\item $\sigma,i\vDash\pi$ if and only if $\pi\in A_i$;
\item $\sigma,i\vDash\neg\phi$ if and only if $\sigma,i\nvDash\phi$;
\item $\sigma,i\vDash\phi_1\vee \phi_2$ if and only if $\sigma,i\vDash\phi_1$ or $\sigma,i\vDash\phi_2$;
\item $\sigma,i\vDash\bigcirc \phi$ if and only if $\sigma,i+1\vDash \phi$;
\item $\sigma,i\vDash\phi_1\U\phi_2$ if and only if there exists $j\ge i$ such that $\sigma,j\vDash\phi_2$ and $\sigma,k\vDash\phi_1$ for all $i\le k<j$;
\end{itemize}
We write $\sigma\vDash\phi$, and say $\sigma$ satisfies $\phi$, if $\sigma,0\vDash\phi$. An execution $\rho$ of a transition system $\T$ is said to satisfy an LTL formula $\phi$, written as $\rho\vDash\phi$, if and only if its trace $\text{Trace}(\rho)\vDash\phi$. Given a control strategy $s$ for $\T$, if all $s$-controlled executions of $\T$ satisfy $\phi$, we write $(\T,s)\vDash \phi$. If such a control strategy $s$ exists, we also say that $\phi$ is \emph{realizable} for $\T$. 

\begin{rem}\em
For technical reasons, we assume that all $\ltlx$ formulas have been transformed into positive normal form \cite[Chapter 5]{baier2008principles}, where all negations appear only in front of the atomic propositions and only the following operators are allowed $\wedge$, $\vee$, $\bigcirc$, $\mathcal{U}$, and $\mathcal{W}$ (defined by $\phi \mathcal{W}\psi=(\phi \U\psi )\vee \always\phi$. We further assume that all negations of atomic propositions are replaced by new atomic propositions.  
\end{rem}

\begin{definition}\em \label{def:implement}
Given an abstraction relation $\alpha$ from $\T_1$ to $\T_2$ and a control strategy $\mu_i$ for $\T_i$ ($i=1,2$), 
$\mu_1$ is called \emph{$\alpha$-implementation} of $\mu_2$, if, for each $n\ge 0$,
$$u_n=\mu_1(x_0,x_1,x_2,\cdots,x_n)$$ is chosen according to 
$$a_n=\mu_2(q_0,q_1,q_2,\cdots,q_n)$$ 
in such a way (as guaranteed by Definition \ref{def:ra} for $\T_1\preceq_{\alpha} \T_2$) that 
$$
\alpha(\text{Post}_{\T_1}(x,u_n))\subset \text{Post}_{\T_2}(q_n,a_n)
$$
for all $x\in \alpha^{-1}(q_n)$, where $q_n\in \alpha(x_n)$.  
\end{definition}

We end this section by stating a soundness result for abstractions.

\begin{thm}\label{thm:sound}
Suppose that $\alpha$ is an abstraction from $\T_1$ to $\T_2$, i.e., $\T_1\,\preceq_{\alpha}\T_2$ and 
and let $\phi$ be an LTL formula. If there exists a control strategy $\mu_2$ for $\T_2$ such that $(\T_2,\mu_2)\vDash\phi$, then there exists a control strategy $\mu_1$, which is an $\alpha$-implementation of $\mu_2$, for $\T_1$ such that $(\T_1,\mu_1)\vDash\phi$. 
\end{thm}

\begin{proof}
Let 
$$\T_1=(Q_1,A_1,R_1,\Pi,L_1)$$ 
and 
$$\T_2=(Q_2,A_2,R_2,\Pi,L_2).$$ 
We show that, by Definitions \ref{def:ra} and \ref{def:implement}, a $\mu_1$-controlled path of $\T_1$ always leads to a $\mu_2$-controlled path of $\T_2$. Suppose we start with $x_k\in Q_1$ and let $q_k$ be arbitrarily chosen from $\alpha(x_k)$, where $k\ge 0$. Suppose $a_k=\mu_2(q_0,q_1,q_2,\cdots,q_k)$ and $u_k=\mu_1(x_0,x_1,x_2,\cdots,x_k)$. Since 
$
\alpha(\text{Post}_{\T_1}(x_k,u_k))\subset \text{Post}_{\T_2}(q_k,a_k),
$
we know that for any $q_{k+1}\in \alpha(x_{k+1})$ and $x_{k+1}\in \text{Post}_{\T_1}(x_k,u_k)$, we have $q_{k+1}\in  \text{Post}_{\T_2}(q_k,a_k)$. This implies that $(q_k,a_k,q_{k+1})$ is a valid transition in $\T_2$ and therefore, by induction, $q_0q_1q_2\cdots$ is a $\mu_2$-controlled path of $\T_2$, if $x_0x_1x_2\cdots$ is a $\mu_1$-controlled path of $\T_1$. Furthermore, by Definitions \ref{def:ra}, we have $L_2(q_k)\subset L_1(x_k)$ for all $k\ge 0.$ Since the trace of $q_0q_1q_2\cdots$ satisfies $\phi$, we know that the trace of $x_0,x_1,x_2\cdots$ also satisfies $\phi$. 
\end{proof}

Based on the proof, it is clear that an abstraction relation preserves not only temporal logic specifications but also linear-time properties in general, because we essentially proved that the controlled traces of $\T_1$ are included in the controlled traces of $\T_2$ (in fact, trace inclusion is equivalent to preservation of LT properties \cite[Theorem 3.15]{baier2008principles}). 

\section{Robust Decidability of Discrete-time Control Synthesis}\label{sec:main}

In this section, we investigate robust abstractions of discrete-time nonlinear systems modelled by nonlinear difference equations with inputs. We establish computational procedures for constructing sound and approximately complete robust abstractions for this class of control systems under very mild conditions. 

\subsection{Perturbed discrete-time control systems as uncertain transition systems}

A \emph{discrete-time control system} is modelled by a difference equation of the form 
\begin{equation}\label{eq:dts}
x(t+1)=f(x(t),u(t)),
\end{equation}
where $x(t)\in X\subset\Real^n$, $u(t)\in U\subset\Real^m$, and $f:\,\Real^n\times\Real^m\ra \Real^n$. 

A \emph{solution} to (\ref{eq:dts}) is an alternating sequence of states and control inputs of the form 
$$
x(0)u(0)x(1)u(1)x(2)u(2)\cdots,
$$
such that (\ref{eq:dts}) is satisfied. 

A \emph{control strategy} for  (\ref{eq:dts}) is a partial function 
$$\sigma:\,(x(0),\cdots,x(t))\mapsto u(t)$$ 
for all $t=0,1,2,\cdots$, which maps the state history up to time $t$ to the control input $u(t)$ at time $t$.

\begin{definition}\em \label{def:dtsts}
The discrete-time control system (\ref{eq:dts}) can be written as a transition system of the form 
\begin{equation}\label{ts}
\S=(Q_\S, A_\S, R_\S, \Pi, L_\S)
\end{equation}
by defining 
\begin{itemize}
\item $Q_\S=X\cup \set{X^c}$; 
\item $A_\S=U$;
\item $(x,u,x')\in R_\S$ if and only if one of the following holds: (i) $x'=f(x,u)$ and $x,x'\in X$; (ii) $x'=X^c$ and $f(x,u)\not\in X$; (iii) $x'=x=X^c$; 
\item $\Pi$ is a set of atomic propositions on $Q_\S$ and $\textbf{in}\in \Pi$;
\item $L_\S:\,Q_\S\ra 2^{\Pi}$ is a labelling function satisfying $\textbf{in}\in L_\S(q)$ for $q\neq X^c$ and $\textbf{in}\not\in L_\S(X^c)$.  
\end{itemize}
\end{definition}
The state $X^c$ and label $\textbf{in}$ are introduced to precisely encode if an out-of-domain transition takes place.

We now introduce an uncertainty model for system (\ref{eq:dts}). 

\begin{definition}\em \label{uc}
Consider system (\ref{eq:dts}) subject to uncertainties of the form 
\begin{equation}\label{eq:dtsw}
x(t+1)=f(x(t),u(t))+w(t),
\end{equation}
where $w(t)\in \delta \mathbb{B}$ for some $\delta\ge 0$. Define $\Delta_\delta$ to consist of transitions $(x,u,x')\not\in R_\S$ such that one of the following holds:  (i) $x'\in f(x,u)+\delta \mathbb{B}$ and $x,x'\in X$; (ii) $x'=X^c$ and $f(x,u)+w\not\in X$ for some $w\in \delta \mathbb{B}$. 
\end{definition}

Clearly, $\S\oplus\Delta_{\delta}$ defined together by Definitions \ref{def:dtsts} and \ref{uc} exactly models (\ref{eq:dtsw}) as summarized in the following proposition. 

\begin{prop}
Each solution of (\ref{eq:dtsw}) that stays in $X$ is an execution of $\S\oplus\Delta_{\delta}$. Conversely, each execution of  
$\S\oplus\Delta_{\delta}$ that stays in $X$ is also a solution of (\ref{eq:dtsw}). 
\end{prop}

\begin{proof}
This is straightforward to verify. Denote  
$$\rho=x(0)u(0)x(1)u(1)x(2)u(2)\cdots.$$ 
If $\rho$ is a solution of (\ref{eq:dtsw}) such that $x(t)\in X$ for all $t\ge 0$. Then there exists $w(0)w(1)\cdots$ such that 
$
x(t+1)= f(x(t),u(t))+w(t),
$
where $w(t)\in\delta \mathbb{B}$ for all $t\ge 0$, which implies that $(x(t),u(t),x(t+1))\in R_{\S}\cup \Delta_{\delta}$. Thus $\rho$ is also an execution of $\S\oplus\Delta_{\delta}$. Now suppose that $\rho$ is an execution  of $\S\oplus\Delta_{\delta}$ such that $x(t)\in X$ for all $t\ge0$. Then $
x(t+1)= f(x(t),u(t))+w(t),
$
where $w(t)\in\delta \mathbb{B}$ for all $t\ge 0$. This shows that $\rho$ is a solution of (\ref{eq:dtsw}). 
\end{proof}

Because of this proposition, in the sequel, when proving soundness results, we always assume that out-of-domain solutions and paths are taken care of by enforcing the solutions and paths to stay in the domain through a safety specification, i.e., by including 
$\always (\textbf{in})$ in the specification. 

\subsection{Soundness of robust abstractions for discrete-time control systems}

\begin{coro}\label{thm:sound2}
Suppose there exists a transition system  $\T$ such that $\S\oplus\Delta_{\delta}\preceq_{\alpha} \T$, where $\S$ and $\Delta_{\delta}$ are defined by Definitions \ref{def:dtsts} and \ref{uc}. Let $\phi$ be an LTL formula over $\Pi$. If there exists a control strategy $\mu$ for $\T$ such that $(\T,\mu)\vDash\phi$, then there exists a control strategy $\kappa$, which is an $\alpha$-implementation of $\mu$, for $\S\oplus\Delta_{\delta}$ such that $(\S\oplus\Delta_{\delta},\kappa)\vDash\phi$. 
\end{coro}

\begin{proof}
It follows directly from Theorem \ref{thm:sound}. 
\end{proof}

It is interesting to note that $(\S\oplus\Delta_{\delta},\kappa)\vDash\phi$ implies that solutions of (\ref{eq:dts}) robustly satisfy $\phi$ in terms of not only additive disturbances modelled by (\ref{eq:dtsw}), but also other types of uncertainties such as measurement errors. To illustrate this, consider a scenario where the controller $\kappa$ is implemented on a system with measurement errors. We assume that this error is bounded, i.e., for each $x(t)\in \Real^n$, its measurement is given by
\begin{equation}\label{eq:error}
\hat{x}(t)= x(t) + e(t), 
\end{equation}
where $e(t)\in \eps \mathbb{B}$ for some $\eps>0$. To make the control strategy $\kappa$ for (\ref{eq:dts}) robust to measurement errors like (\ref{eq:error}), we can simply strengthen the labeling function $L$ of $\S$ as follows. A labelling function $\hat{L}:\,\Real^n\ra 2^\Pi$ is said to be the \emph{$\eps$-strengthening} of another labelling function $L:\,\Real^n\ra 2^\Pi$, if $\pi\in \hat{L}(x)$ if and only if $\pi\in L(y)$ for all $y\in x+\eps\mathbb{B}$.

The remaining technical results of the paper rely on the following assumption. 

\begin{ass}\label{ass}
The function $f:\,\Real^n\times\Real^m$ is locally Lipschitz continuous in both arguments. The sets $X$ and $U$ are compact. 
\end{ass}

The above assumption on $f$ is very mild and is satisfied as long as the function $f:\,\Real^n\times\Real^m$ is differentiable with respect to both variables. 

\begin{prop}
Let $\hat{\S}=(Q, A, R, \Pi, \hat{L})$, which is obtained from $\S$ in Definition \ref{def:dtsts} by replacing $L$ with its $\eps$-strengthening $\hat{L}$. Suppose that the assumptions of Corollary \ref{thm:sound2} hold with $\hat{\S}$ in place of $\S$. Then $(\S,\kappa)\vDash\phi$, subject to measurement errors described in (\ref{eq:error}), provided that $(L+1)\eps\le \delta$, where $L$ is the uniform Lipschitz constant for both variables of $f$ on the compact set $(X+\eps\mathbb{B})\times U$. 
\end{prop}

\begin{proof}
We have  $\hat{\S}\oplus\Delta_{\delta}\preceq_{\alpha} \T$. The goal is to show that, despite the measurement errors, $\kappa$-controlled traces of $\S$ are a subset of the $\kappa$-controlled traces of $(\hat{\S},\Delta)$ and therefore satisfies $\phi$. Starting from $x(0)$, let $\hat{x}(0)$ be the measurement taken for $x(0)$. Suppose that an action $u(0)=\kappa(\hat{x}(0))=\mu(q_0)$ is chosen by $\kappa$, where $q_0\in \alpha(\hat{x}(0))$. Let $L_1$ be the labelling function for $\T$. Then $L_1(q_0)\subset\hat{L}(\hat{x}(0))$ by the definition of the robust abstraction. Since $\hat{L}$ is the {$\eps$-strengthening} of $L$ and $x(0)\in \hat{x(0)}+\eps\mathbb{B}$, it follows that $L_1(q_0)\subset \hat{L}(\hat{x}(0))\subset L(x(0))$.  

We suppose by induction that $L_1(q_k) \subset L(x(k))$ holds for some $k\ge 0$, where $q_k\in \alpha(\hat{x}(k))$ and $\hat{x}(k)\in x(k)+\eps\mathbb{B}$. The action at time $k$ is given by $u(k)=\kappa(\hat{x}(0),\cdots,\hat{x}(k))$, which implements $a_k=\mu(q_0,\cdots,q_k)$ in the sense of Definition \ref{def:implement}. The next state under $u(k)$ is given by $x(k+1)=f(x(k),u(k))$, whose measurement is $\hat{x}(k+1)=x(k+1)+e(k+1)\in x(k+1)+\eps\mathbb{B}$. Hence $L_1(q_{k+1})\subset \hat{L}(\hat{x}(k+1))$ implies $L_1(q_{k+1})\subset  \hat{L}(\hat{x}(k+1)) L(x(k+1))$. Thus, $L(q_k) \subset L(x(k))$ for all $k\ge 0$. 

We show that $(q_k,a_k,q_{k+1})$ is a valid transition in $\T$. Note that 
\begin{align*}
&\hat{x}(k+1)\\
&= x(k+1)+e(k+1) \\
&= f(x(k),u(k)) + e(k+1) \\
&= f(\hat{x}(k),u(k)) + (f(x(k),u(k)) - f(\hat{x}(k),u(k))) + e(k+1). 
\end{align*}
Since $f$ is $L$-Lipschitz continuous in both arguments on the compact set $(X+\eps\mathbb{B})\times U$, the above equation shows that 
$$
\hat{x}(k+1) \in f(\hat{x}(k),u(k)) + (L + 1)\eps  \mathbb{B} \subset f(\hat{x}(k),u(k)) + \delta \mathbb{B},
$$
because $(L + 1)\eps\le\delta$. Hence, by the choice of $u(k)$ by $\kappa$ (which is an $\alpha$-implementation of $\mu$), we have
\begin{align*}
q_{k+1}&\in \alpha(\hat{x}(k+1)) \\
&\subset \alpha(f(\hat{x}(k),u(k)) + \delta \mathbb{B})\\
&\subset \alpha(\text{Post}_{\hat{\S}\oplus\Delta}(\hat{x}(k),u(k)))\\
&\subset \text{Post}_{\T}(q_k,a_k),
\end{align*} 
where $\hat{x}(k)\in \alpha^{-1}(q_k)$, which shows that $(q_k,a_k,q_{k+1})$ is a valid transition in $\T$ and therefore $q_0q_1q_2\cdots$ is a valid path for $\T$. Since the trace of this path satisfies $\phi$ and $L_1(q_k)\subset L(x(k))$ for all $k\ge 0$, it follows that the trace of $x(0)x(1)x(2)\cdots$ also satisfies $\phi$.  
\end{proof}

\begin{rem}\em \label{rem:measure}
The soundness result above states that to cope with measurement errors, we only need to choose $\delta$ sufficiently large such that $(L+1)\eps\le \delta$ and strengthen the labelling function by a factor of $\eps$. This condition  simplifies the two robustness margins $(\gamma_1,\gamma_2)$ considered in the work \cite{liu2014abstraction,liu2016finite} and also does not require that the abstraction relation to be non-deterministic in order to be robust with respect to measurement errors as stated in \cite[Section VI.6]{reissig2016feedback}. 
\end{rem}

\subsection{Approximate completeness of robust abstractions for discrete-time control systems}\label{sec:complete}

In this section, we show that, under Assumption \ref{ass}, computing robust abstractions for the discrete-time control system (\ref{eq:dts}) is approximately complete, in the sense that, for arbitrary numbers $0\le \delta_1<\delta_2$, we can find a finite transition system $\T$ such that $\S\oplus\Delta_{\delta_1}\preceq \T\preceq \S\oplus\Delta_{\delta_2}$, where $\S$ and $\Delta_{\delta_i}$ ($i=1,2$) are defined in Definitions \ref{def:dtsts} and \ref{uc}. This result is made precise by the following theorem, which we present as the main result of the paper.

\begin{thm}\label{thm:complete}
For any numbers $0\le \delta_1<\delta_2$, let $\Delta_{\delta_i}$ ($i=1,2$) be given by Definition \ref{uc} with $\delta=\delta_i$. For any numbers $0\le \eps_1<\eps_2$, let $L_{\S_i}$ ($i=1,2$) be the $\eps_i$-strengthening of $L_\S$. Let 
$$
\S_{i} = (Q_\S,A_\S,R_\S\cup \Delta_{\delta_i},\Pi,L_{\S_i}),\quad i=1,2. 
$$
Then there exists a finite transition system $\T$ such that 
\begin{equation}\label{eq:main}
\S_{1}\preceq \T \preceq \S_2. 
\end{equation}
\end{thm}

To prove Theorem \ref{thm:complete}, we need the following lemma on over-approximation of the reachable set of a box in $\Real^n$ under a nonlinear map. 
\begin{lem}\label{lem:reach}
Fix any $\delta>0$, any box (also called an interval or a hyperrectangle) $[x]\subset \Real^n$, and any $u\in U$. For all $\eps>0$, there exists a finitely terminated algorithm to compute an over-approximation of the reachable set of $[x]$ under (\ref{eq:dtsw}), i.e., the set
$$
\text{Reach}_{(\ref{eq:dtsw})}([x],u)=f([x],u)+\delta \mathbb{B},
$$
 such that 
$$
\text{Reach}_{(\ref{eq:dtsw})}([x],u) \subset \widehat{\text{Reach}}_{(\ref{eq:dtsw})}([x],u)\subset \text{Reach}_{(\ref{eq:dtsw})}([x],u) +\eps\mathbb{B},
$$
where $\widehat{\text{Reach}}_{(\ref{eq:dtsw})}([x],u)$ is the computed over-approximation given as a union of boxes. 
\end{lem}
\begin{proof}
This is a well-known result in interval analysis, known as outer approximation of the image set of a function. It can be proved, for example, using the results in  \cite[Chapter 3]{jaulin2001applied}. 
Here we include a proof for completeness. Let $\mathbb{IR}^n$ denote the set of all boxes in $\Real^n$. Let $[f_u]:\,\mathbb{IR}^n\rightarrow\mathbb{IR}^m$ be a \emph{convergent inclusion function} \cite{jaulin2001applied} of $f(\cdot,u)$, which satisfies the following two conditions: 
\begin{itemize}
\item $f([y],u)\subset [f_u]([y])$ for all $[y]\in \mathbb{IR}^n$; 
\item $\lim_{w([y])\to 0}w([f_u]([y]))=0$,
\end{itemize}
where $w([y])$ is the width of $[y]$, given by $\max_{1\leq i\leq n}\{\overline{y_i}-\underline{y_i}\}$ if we write  
$[y]=[y_1]\times\cdots\times[y_n]\subset \Real^n$ and $[y_i]=[\underline{y}_i,\overline{y}_i]\subset \Real$ for $i=1,\cdots,n$. Without loss of generality, assume that $\eps<1$. Because $f$ is $L$-Lipschitz continuous on $[x]$ for some $L>0$, we can find an inclusion function such that $w([f_u]([y]))\le Lw([y])$ for any subintervals of $[x]$.  
We mince the interval $[x]$ into subintervals such that the largest width of among these subintervals is smaller than $\frac{\eps}{2L}$. For each such interval $[y]$, we evaluate $[f_u]([y])$ and obtain the interval $[z]=[f_u]([y])+\delta\mathbb{B}$. Let $\mathcal{Y}$ denote the collection of all such intervals\footnote{Such a collection $\mathcal{Y}$ is called a non-regular paving of $\Real^n$, which can be regularized \cite[Chapter 3]{jaulin2001applied} to reduce the number of boxes and hence reduce complexity, but this is not necessary for our purpose.} and let $Y$ be its union. We claim that $$Y=\widehat{\text{Reach}}_{(\ref{eq:dtsw})}([x],u)$$ satisfies the requirement of this lemma. This is clearly true because, for each interval $[z]=[f_u]([y])+\delta\mathbb{B}$, we have $f([y])+\delta\mathbb{B}\subset [z]$ and the  distance from $[z]$ to the true reachable set $\text{Reach}_{(\ref{eq:dtsw})}([x],u)$ is bounded by $w([f_u]([y]))\le L\cdot w([y])\le \frac{\eps}{2}$. The proof for Lemma \ref{lem:reach} is also summarized in pseudo code format in Algorithm \ref{alg:reach}. \end{proof}

\begin{algorithm}[htbp]
  \caption{Computation of an over-approximation of $\text{Reach}_{(\ref{eq:dtsw})}([x],u)$ (Lemma \ref{lem:reach})}
  \label{alg:reach}
  \begin{algorithmic}[1]
  \Require $[x]$, $\delta$, $\eps>0$, the Lipschitz constant $L$ for $f(\cdot,u)$, and a centred convergent inclusion function $[f_u]$ for $f(\cdot,u)$
	\State $List\leftarrow [x]$
	\State $\mathcal{Y}\leftarrow \emptyset$
    \While{$List\neq \varnothing$}
    \State $[y]\leftarrow First(List)$
    \State $List \leftarrow List\setminus \set{[x]}$    
    \If{$w([y])\le \frac{\eps}{2L}$}
    \State $[z]\leftarrow [f_u]([y])+\delta\mathbb{B}$
    \State $\mathcal{Y}\leftarrow \mathcal{Y}\cup \set{[z]}$
    \Else
    \State $\{Left[y],Right[y] \}=Bisect([y])$
    \State $List\leftarrow List\cup\set{Left[y],Right[y]}$
    \EndIf
    \EndWhile
    \State $Y\leftarrow \cup_{[z]\in\mathcal{Y}}[z]$
    \State  \Return $Y=\widehat{\text{Reach}}_{(\ref{eq:dtsw})}([x],u)$
   \end{algorithmic}
\end{algorithm}

\begin{proof}[of Theorem \ref{thm:complete}]
The proof is constructive and we construct a finite transition system 
$$\T=(Q_\T, A_\T, R_\T, \Pi, L_\T)$$ 
as follows. 

For a positive integer $k$, let $\Z^k$ denote the $k$-dimensional integer lattice, i.e., the set of all $k$-tuples of integers. For parameters $\eta>0$ and $\mu>0$ (to be chosen later), define
$$
[\Real^n]_{\eta}=\eta\Z^n,\quad [\Real^m]_{\mu}=\mu\Z^m, 
$$
where $\mu\Z^k=\set{\mu z:\,z\in\Z^k}$ (for $k=n,m$). Define a relation $\alpha$ from $Q_\S$ to $[\Real^n]_{\eta}\cup\set{X^c}$ by 
$$\set{(x,q):\,q=\eta\floor{\frac{x}{\eta}},x\in X}\cup\set{(X^c,X^c)},$$ 
where $\floor{\cdot}$ is the floor function (i.e., $\floor{x}=(\floor{x_1},\cdots,\floor{x_n})$ and $\floor{x_i}$ gives the largest integer less than or equal to $x_i$). Let $Q_\T$ be $\alpha(Q_\S)$ and $A_\T=\set{a:\,\exists u\in A_{\S}\text{ s.t. }a=\mu\floor{\frac{u}{\mu}}}$ (which are both non-empty by definition and are finite because $X$ and $U$ are compact). Note that this gives a deterministic relation in the sense that $\alpha(x)$ is single-valued for all $x$. It is straightforward to verify that
\begin{equation}\label{eq:alpha}
\alpha^{-1}(\alpha(B))\subset B+\eta\mathbb{B},
\end{equation}
for any set $B\subset \Real^n \cup{X^c}$, with the slight abuse of notation that $X^c+x=X^c$ for any $x\in\Real^n$.

We next construct $R_\T$. For each $q\in Q_\T$ and $a\in A_\T$, denote by 
$$
\text{Reach}_{\S_1}(\alpha^{-1}(q),a)= \bigcup_{x\in\alpha^{-1}(q)}\text{Post}_{\S_2}(x,a). 
$$
We let $(q,a,q')$ be included in $R_\T$ if and only if 
$$
q'\in \alpha(\widehat{\text{Reach}}_{\S_1}(\overline{\alpha^{-1}(q)},a)),  
$$
i.e., 
\begin{equation}\label{eq:post}
\text{Post}_{\T}(q,a) = \alpha(\widehat{\text{Reach}}_{\S_1}(\overline{\alpha^{-1}(q)},a)), 
\end{equation}
where 
$
\widehat{\text{Reach}}_{\S_1}(\overline{\alpha^{-1}(q)},a)
$
is computed from Lemma \ref{lem:reach} by setting $[x]=\overline{\alpha^{-1}(q)}$, $u=a$, and $\delta=\delta_1$. In particular, we set 
$
\widehat{\text{Reach}}_{\S_1}(\overline{\alpha^{-1}(q)},a)= \widehat{\text{Reach}}_{(\ref{eq:dtsw})}([x],u), 
$
if $\widehat{\text{Reach}}_{(\ref{eq:dtsw})}([x],u)\subset X$, and 
$$
\widehat{\text{Reach}}_{\S_1}(\overline{\alpha^{-1}(q)},a)= \widehat{\text{Reach}}_{(\ref{eq:dtsw})}([x],u)\cup{\set{X^c}},
$$
if $\widehat{\text{Reach}}_{(\ref{eq:dtsw})}([x],u)\not\subset X$. 

Then it follows from Lemma \ref{lem:reach} that
\begin{align*}
\alpha(\bigcup_{x\in\alpha^{-1}(q)}\text{Post}_{\S_1}(x,a)) &\subset \alpha(\text{Reach}_{\S_1}(\overline{\alpha^{-1}(q)},a))\\
& \subset \alpha(\widehat{\text{Reach}}_{\S_1}(\overline{\alpha^{-1}(q)},a))\\
&= \text{Post}_{\T}(q,a), 
\end{align*}
which verifies condition (ii) of Definition \ref{def:ra} for $\S_1\preceq_{\alpha} \T$.  

Consider $\alpha^{-1}$ as a relation from $Q_\T$ to $Q_\S$. Then for each $x\in Q_\S$ and $u\in A_\S$, we can choose $a=\mu\floor{\frac{u}{\mu}}\in A_\T$ such that  
\begin{align*}
\alpha^{-1}(\bigcup_{q\in\alpha(x)}\text{Post}_{\T}(q,a)) &= \alpha^{-1}(\text{Post}_{\T}(q,a))\\
&\subset \alpha^{-1}(\alpha(\widehat{\text{Reach}}_{\S_1}(\overline{\alpha^{-1}(q)},a)))\\
& \subset \widehat{\text{Reach}}_{\S_1}(\overline{\alpha^{-1}(q)},a)+\eta\mathbb{B}\\
& \subset \text{Reach}_{\S_1}(\overline{\alpha^{-1}(q)},a) +(\eta+\eps)\mathbb{B}. 
\end{align*}
where we used (\ref{eq:post}), (\ref{eq:alpha}), and Lemma \ref{lem:reach}.  We claim that, if we can choose $\eta$, $\mu$, and $\eps$ sufficiently small such that 
\begin{equation}\label{eq:small}
\delta_1+L(\eta+\mu)+\eta+\eps\le \delta_2,
\end{equation}
then 
\begin{equation}\label{eq:goal}
\text{Reach}_{\S_1}(\overline{\alpha^{-1}(q)},a) +(\eta+\eps)\mathbb{B}\subset \text{Post}_{\S_2}(x,u).
\end{equation}
Note that $\overline{\alpha^{-1}(q)}\subset x+\eta \mathbb{B}$ and $a\in u+\mu\mathbb{B}$. We first assume that $X^c\not \in \text{Reach}_{\S_1}(\overline{\alpha^{-1}(q)},a)$. Without loss of generality, we can assume that $\eta\le 1$ and $\mu\le 1$. Because $f$ is Lipschitz continuous in both arguments on the compact set $(X+\mathbb{B})\times(U+\mathbb{B})$ (we use $L$ to indicate the uniform Lipschitz constant for both variables on this set), it follows that 
\begin{equation*} 
\text{Reach}_{\S_1}(\overline{\alpha^{-1}(q)},a)\subset f(x,u)+[\delta_1+L(\eta+\mu)]\mathbb{B}.
\end{equation*}
Combining the displayed equations above, we obtain 
\begin{align*} 
\alpha^{-1}(\bigcup_{q\in\alpha(x)}\text{Post}_{\T}(q,a))&\subset f(x,u)+\delta_2\mathbb{B}\\
&= \text{Post}_{\S_2}(x,u), 
\end{align*}
which verifies condition (ii) of Definition \ref{def:ra} for $\T\preceq_{\alpha}\S_2$, because $X^c\in \alpha^{-1}(\bigcup_{q\in\alpha(x)}\text{Post}_{\T}(q,a))$ would also imply $X^c\in \text{Post}_{\S_2}(x,u)$.   

Now we define $L_\T$. For each $q\in Q_\T$, define 
$$
\pi\in L_\T(q) 
$$
if and only if $\pi\in L_\S(x)$ for all $x\in q+\frac{\eps_1+\eps_2}{2}\mathbb{B}$.  Choose $\eta$ sufficiently small such that $\eta+\frac{\eps_1+\eps_2}{2}<\eps_2$. This is possible because $\eps_2>\eps_1$. To verify condition (iii) of Definition \ref{def:ra} for  $\S_1\preceq_{\alpha} \T$ and $\T\preceq_{\alpha^{-1}}\S_2$, we need to check that 
\begin{equation}\label{eq:label1}
L_{\S_2}(x) \subset L_\T(q) 
\end{equation}
and 
\begin{equation}\label{eq:label2}
L_\T(q) \subset L_{\S_1}(x) 
\end{equation}
for all $(x,q)\in\alpha$. Fix any $(x,q)\in \alpha$. If $\pi\in L_{\S_2}(x)$, then $\pi\in L_\S(y)$ for all $y\in x+\eps_2\mathbb{B}$. Since $q+\frac{\eps_1+\eps_2}{2}\mathbb{B}\subset x+[\eta+\frac{\eps_1+\eps_2}{2}]\mathbb{B}\subset x+\eps_2\mathbb{B}$, we have $\pi\in L_{\S}(y)$ for all $y\in q+\frac{\eps_1+\eps_2}{2}\mathbb{B}$ and $\pi\in L_\T(q)$. Hence, (\ref{eq:label1}) holds. If $\pi\in L_{\T}(q)$, then $\pi\in L_\S(y)$ for all $y\in q+\frac{\eps_1+\eps_2}{2}\mathbb{B}$ by the definition of $L_\T$. Since $x+\eps_1\mathbb{B}\subset q+(\eta+\eps_1)\mathbb{B}\subset q+\frac{\eps_1+\eps_2}{2}\mathbb{B}$, we have $\pi\in L_\S(y)$ for all $y\in x+\eps_1\mathbb{B}$ and $\pi\in L_{\S_1}(x)$. Hence, (\ref{eq:label2}) holds. 

We have verified $\S_1\preceq \T \preceq \S_2$ by checking all the conditions of Definition \ref{def:ra}. The main steps of the proof are also summarized in pseudo code format in Algorithm \ref{alg:abs}. 
\end{proof}

\begin{algorithm}
\caption{Computation of an approximately complete robust abstraction $\mathcal T$ for $\S$ (Theorem \ref{thm:complete})}
\label{alg:abs}
\begin{algorithmic}[1]
\Require $\S=(Q_\S, A_\S, R_\S, \Pi, L_\S)$, numbers $0\le\delta_1<\delta_2$ and $0\le\eps_1<\eps_2$
\State Set $L_{\S_i}$ be the $\eps_i$-strengthening of $L_\S$ ($i=1,2$)
\State Set $\Delta_{\delta_i}$ according to Definition \ref{uc} ($i=1,2$)
\State Set $\S_i = (Q_\S,A_\S,R_\S \cup \Delta_{\delta_i},\Pi,L_{\S_i})$  ($i=1,2$)
\State Choose rational numbers $\eta\in (0,1)$ and $\eps\in (0,1)$ such that $\delta_1+L(\eta+\mu)+\eta+\eps\le \delta_2$ and $\eta+\frac{\eps_1+\eps_2}{2}<\eps_2$, where $L$ is the uniform Lipschitz constant of $f$ on the compact set $(X+\mathbb{B})\times(U+\mathbb{B})$
\State Set $Q_\T=\set{x\in[\Real^n]_{\eta}:\,\exists x\in Q_{\S}\text{ s.t. }x=\eta\floor{\frac{u}{\eta}}}\cup\set{X^c}$ 
\State Set $A_\T=\set{a\in [\Real^m]_{\mu}:\,\exists u\in A_{\S}\text{ s.t. }a=\mu\floor{\frac{u}{\mu}}}$
\ForAll{$q\in Q_\T$}
\State $L_\T(q)\leftarrow\emptyset$
\ForAll{$\pi\in \Pi$}
\If{$\pi\in L_\S(x)$ for all $x\in q+\frac{\eps_1+\eps_2}{2}\mathbb{B}$}
\State  $L_\T(q)\leftarrow L_\T(q)\cup\set{\pi}$
\EndIf
\EndFor
\EndFor
\State $R_\T\leftarrow\emptyset$
\ForAll{$q \in Q_\T$}
\ForAll{$a \in A_\T$}
\If{$q'\in \alpha(\widehat{\text{Reach}}_{(\ref{eq:dtsw})}(\overline{\alpha^{-1}(q)},a))$}
\State $R_\T \leftarrow R_\T\cup \set{(q,a,q')}$
\EndIf
\EndFor
\EndFor
\State \Return ${\mathcal T}= (Q_\T, A_\T, R_\T, \Pi, L_\T)$ 
\end{algorithmic}
\end{algorithm}

\begin{rem}\em
While the disturbance sets are so chosen for simplicity of presentation, they do not have to be of the form $\delta\mathbb{B}$. In fact, if we choose two arbitrary sets $W_1$ and $W_2$ in place of $\delta_1\mathbb{B}$ and $\delta_2\mathbb{B}$ in Definition \ref{uc} such that there exists $\eps>0$ such that $W_1+\eps\mathbb{B}\subset W_2$, then a completeness result similar to Theorem \ref{thm:complete} can be stated. Furthermore, $\delta$ can be a vector in $\Real^n$ instead of a scalar, in which case $\delta_i\mathbb{B}$ becomes a hyperrectangle and the condition $0\le \delta_1<\delta_2$ is a componentwise inequality. 
\end{rem}

\begin{rem}\em
In the proof of Theorem \ref{thm:complete}, we in fact construct a single-valued abstraction relation $\alpha$. While the main results of the paper are presented for the case where $\alpha$ can be multi-valued, it appears, in view of the proof of Theorem \ref{thm:complete}, that for practice purposes, $\alpha$ may always be chosen to be deterministic, while still preserving robustness (see also Remark \ref{rem:measure}). 
\end{rem}

Finally, we would like to point out that Theorem \ref{thm:complete} shows that there exists an \emph{approximately complete} abstraction procedure for discrete-time nonlinear control systems of the form (\ref{eq:dts}) in the sense that, if a specification $\phi$ is realizable for $\S_2$ (namely, a $\delta_2$-perturbation of $\S$), then there is a robust abstraction $\T$ of $\S_1$, which is a $\delta_1$-perturbation of $\S$, such that $\phi$ is realizable for $\T$ and hence it is also realizable for $\S_1$. Note that $\S_1$ and $\S_2$ can be made arbitrarily close by choosing $\delta_2$ close to $\delta_1$ and $\eps_2$ close to $\eps_1$. Since the proof of above theorem is constructive, we can algorithmically synthesize a control strategy for $\S_1$ by computing $\T$ first and then solving a discrete synthesis problem for $\T$ with the specification $\phi$. 
We summarize this in the following corollary.

\begin{coro}
Let $\S_1$, $\S_2$, and $\phi$ be as defined in Theorem \ref{thm:complete}. There is a decision procedure to answer one of the following two questions: 
\begin{itemize}
\item[(i)] there exists a control strategy $\kappa$ (and one can algorithmically construct it) such that $(\S_1,\kappa)\vDash\phi$; 
\item[(ii)] $\phi$ is not realizable for $\S_2$. 
\end{itemize}
\end{coro}

\section{An example}\label{sec:ex}

We use a simple motion planning example to illustrate our results. Consider a vehicle steering problem, where the dynamics of the vehicle are given by the so-called bicycle model \cite{astrom2010feedback}. The same example is used for illustration of abstraction-based control design in  \cite{reissig2016feedback,zamani2012symbolic,rungger2016scots}. The model is given by 
  \begin{equation*}
        \begin{bmatrix}
      \dot{x}_1\\
      \dot{x}_2\\
      \dot{x}_3
    \end{bmatrix} = 
    \begin{bmatrix}
      u_1 \cos (\alpha+x_3)/\cos (\alpha)\\
      u_1 \sin (\alpha+x_3)/\cos (\alpha)\\
      u_1 \tan (u_2)
    \end{bmatrix}
  \end{equation*}
where $(x_1,x_2,x_3)=(x,y,\theta)$ and $(u_1,u_2)=(v,\phi)$. The constant $b=1$ is the wheel base and $a=0.5$ is the distance between centre of mass and rear wheels. The states consist of the coordinates of the centre of the mass $(x,y)$ and the heading angle $\theta$. The controls consist of the wheel speed $v$ and the steering angle $\phi$. The variable $\alpha$ is the angle of velocity depending on $\phi$. 

Let $X=[7,10]\times [0,4.5]\times [-\pi,\pi]$ and $U=[-1,1]\times [-1,1]$. Consider a workspace and a specification given by
    \begin{equation*}
      \varphi=A_I\wedge \square( \neg A_O) \wedge \lozenge A_G,
    \end{equation*}
    where
    \begin{equation*}
      \begin{aligned}
        A_I&=[7.6,0.4,\pi/2]^T,\\
        A_G&=[9,9.6]\times[0,0.6]\times[-\pi,\pi],\\
        A_O&=A_{O1}\cup A_{O2} \cup A_{O3},\\
        A_{O1}&=[8.2,8.4]\times[0,3.6]\times[-\pi,\pi],\\
        A_{O2}&=[8.4,9.4]\times[3.4,3.6]\times[-\pi,\pi],\\
        A_{O3}&=[9.4,10]\times[2.4,2.6]\times[-\pi,\pi]. 
      \end{aligned}
    \end{equation*}

To design a control strategy to realize this specification, we discretize the model using a sampling time step $\tau=0.3$. We first consider the case with no disturbance, i.e., $\delta=0$. Using the discretization parameters $\eta=0.2$ and $\mu= 0.3$, the resulting nominal abstraction consists of 12,880 states and 3,023,040 transitions. The computation time was $7.3$s for computing the abstraction and $8.6$s for solving the synthesis problem on a 2.2GHz Intel Core i7 processor. A feasible trajectory is shown in Figure \ref{fig:sim1}. To design a robust control strategy, we consider an additive disturbance of size $\delta=0.05$ on the right-hand side of the system. We compute a robust abstraction by setting $\delta_1=0.05$ and $\eta=0.05$. The resulting robust abstraction consists of {$782,691$} states and $1,727,548,752$ transitions. The computation time was {$2,327$s} for abstraction  and {$2,289$s} for synthesis on the same processor. A feasible trajectory is shown in Figure \ref{fig:sim2}. Using the same controller, a simulated trajectory with an additive disturbance of size $\delta=0.15$ is shown to violate the specification. Furthermore, Theorem \ref{thm:complete} implies that, for any $0.05\le\delta_1<\delta_2$, by further refining the abstraction, we should be able to assert that either the specification is robustly realizable with a disturbance of size $\delta_1$ or the specification is not realizable with a disturbance of size $\delta_2$. 

    \begin{figure}[ht]
      \centering
           \includegraphics[scale=0.53]{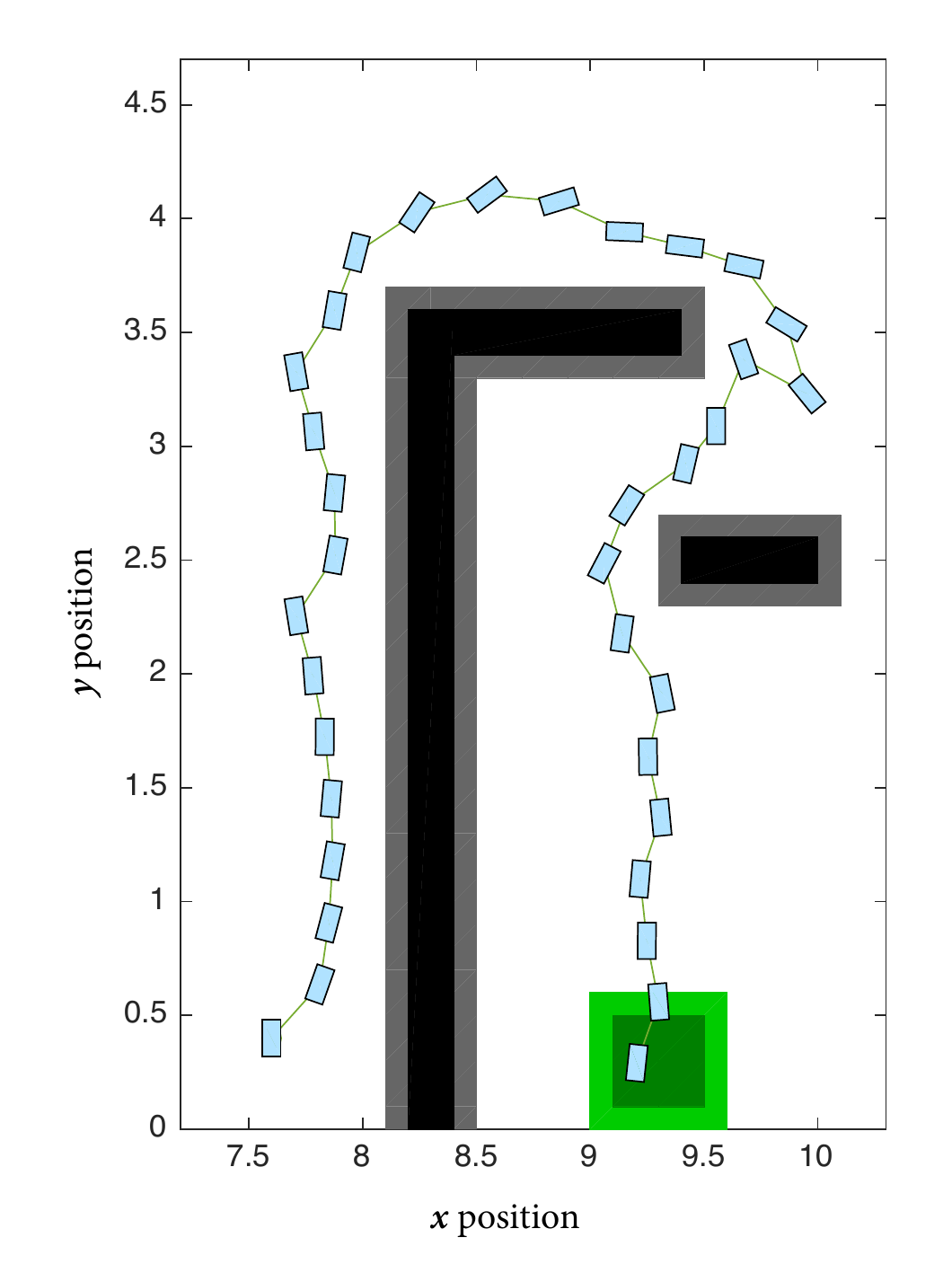}
      \label{fig:sim1}
     \caption{A simulated trajectory from a nominal abstraction that satisfies the specification.}
    \end{figure}

    \begin{figure}[ht]
      \centering
      \includegraphics[scale=0.4]{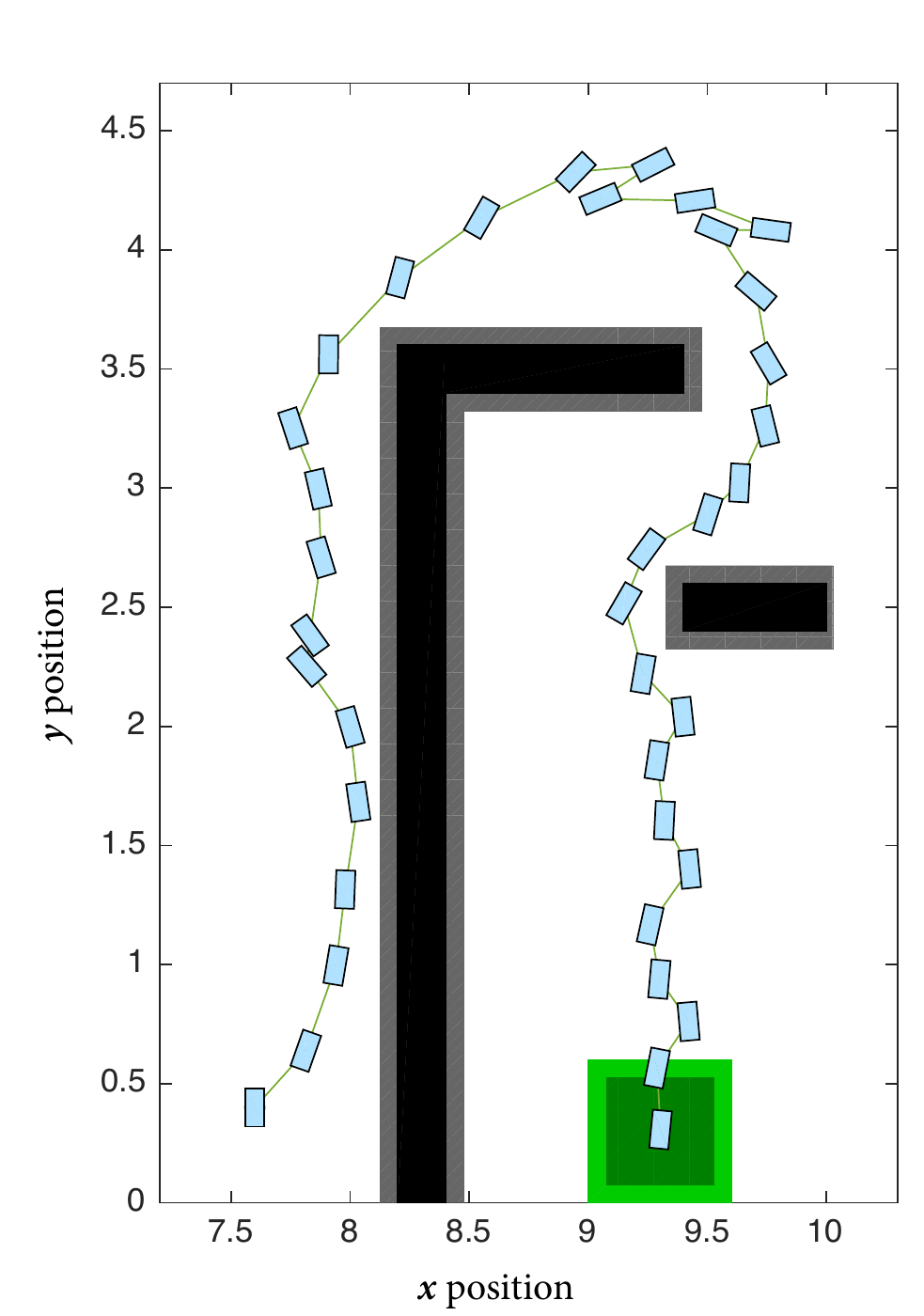}\includegraphics[scale=0.4]{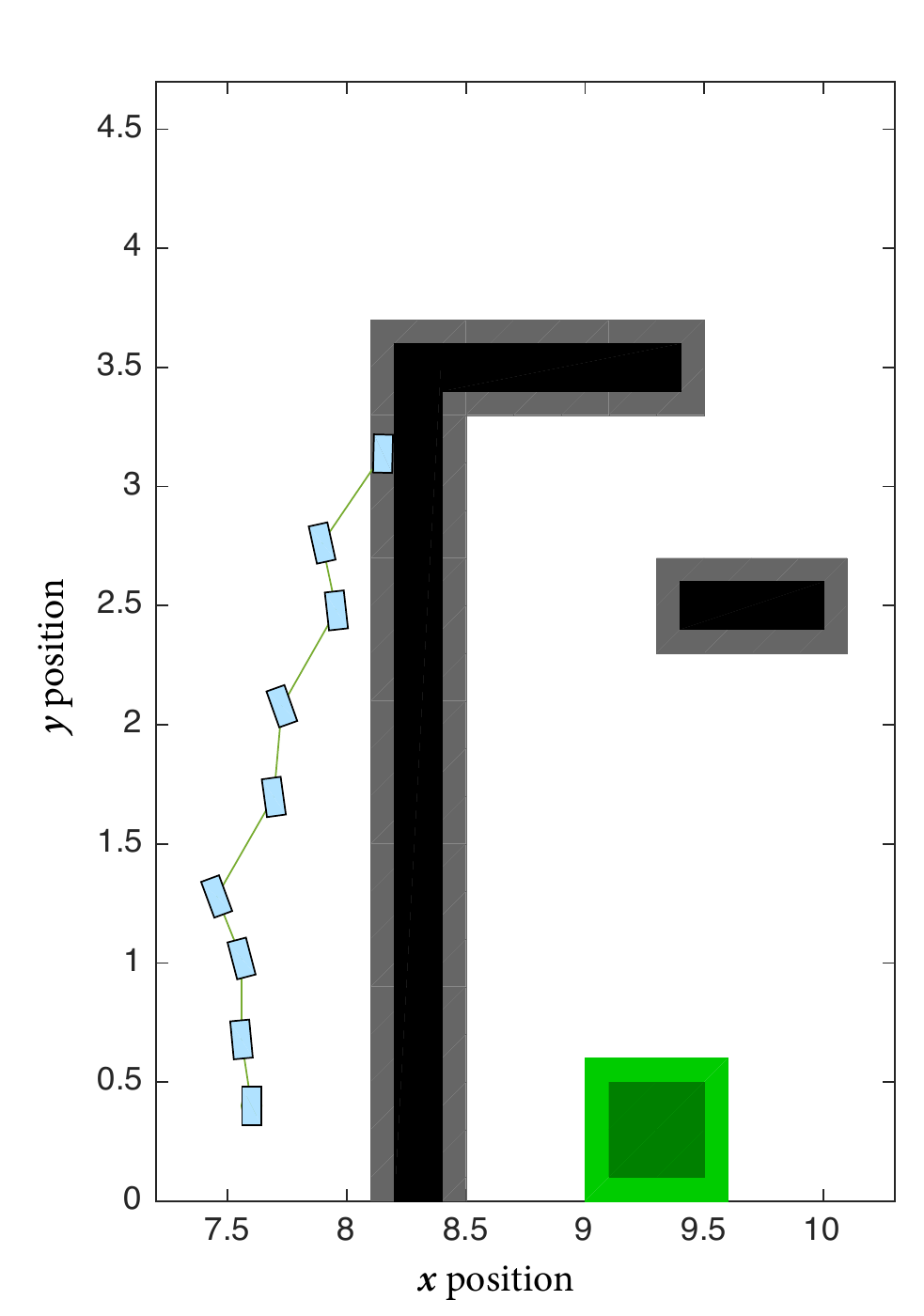}
            \label{fig:sim2}
\caption{A valid trajectory (left) obtained from a robust abstraction with $\delta=0.05$ and a failed trajectory (right) with disturbance size $\delta=0.15$.}
    \end{figure}

\section{Conclusions and Discussions}\label{sec:conc}

We proposed a computational framework for designing robust abstractions for control synthesis. It is shown that robust abstractions are not only sound in the sense that they preserve robust satisfaction of linear-time properties, but also approximately complete in the sense that, given a concrete discrete-time control system and an arbitrarily small perturbation of this system, there exists a finite transition system that robustly abstracts the concrete system and is abstracted by the perturbed system at the same time. Consequently, the existence of controllers for a general discrete-time nonlinear control system and linear-time specifications is robustly decidable: if a specification is robustly realizable, there is a decision procedure to find a (potentially less) robust control strategy. 

It is interesting to note that the connection between robustness and decidability appeared in different contexts. Recently, the notion of $\delta$-decidability for satisfiability over the reals \cite{gao2012delta} and $\delta$-reachability analysis \cite{kong2015dreach} have been proposed to turn otherwise undecidable problems into decidable ones. A notion of ``robustness implies decidability" was proposed in early work in \cite{fraenzle2001what} for verifying bounded properties for polynomial hybrid automaton and in \cite{asarin2001perturbed} for reachability analysis of several simple models of hybrid systems. Finally, the early work in \cite{anderson1975output} showed  that \emph{robust} stability is decidable for linear systems in the context of output feedback stabilization. In this sense, the current work can serve as an example of ``robustness implies decidability" in the context of linear-time logic control synthesis for nonlinear systems.

\section{Acknowledgments}

This research was supported in part by NSERC Canada and the University of Waterloo. The author would like to thank Necmiye Ozay and Yinan Li for stimulating discussions on related topics and the anonymous reviewers for helpful comments and suggestions.

\bibliographystyle{abbrv}
\bibliography{robust_abs} 
\end{document}